\documentclass[a4paper,USenglish,cleveref,autoref,thm-restate]{lipics-v2021}

\nolinenumbers 
\title{Treewidth of Outer $k$-Planar Graphs}
\allowdisplaybreaks[1]

\pdfoutput=1 
\hideLIPIcs  

\author{Rafał Pyzik}
{Institute of Theoretical Computer Science, Faculty of Mathematics and Computer Science, Jagiellonian University, Kraków, Poland}
{rafal.pyzik@student.uj.edu.pl}
{https://orcid.org/0000-0003-4147-7000}
{}

\authorrunning{R. Pyzik}
\Copyright{Rafał Pyzik}

\date{\today}

\acknowledgements{I thank Grzegorz Gutowski for introducing me to this work, for many helpful discussions, and for proofreading this paper. I am also grateful to the reviewers for their insightful comments and suggestions.}


\EventEditors{Vida Dujmovi\'c and Fabrizio Montecchiani}
\EventNoEds{2}
\EventLongTitle{33rd International Symposium on Graph Drawing and Network Visualization (GD 2025)}
\EventShortTitle{GD 2025}
\EventAcronym{GD}
\EventYear{2025}
\EventDate{September 24--26, 2025}
\EventLocation{Norrk\"{o}ping, Sweden}
\EventLogo{}
\SeriesVolume{357}
\ArticleNo{26}

\usepackage{hyperref}
\usepackage{graphicx}
\usepackage[textsize=small]{todonotes}

\definecolor{defblue}{rgb}{0.121,0.47,0.705}
\definecolor{linkblue}{rgb}{0.098,0.098,0.4392}
\let\emph\relax
\DeclareTextFontCommand{\emph}{\color{defblue}\em}

\newcommand{\set}[1]{\left\{#1\right\}}
\newcommand{\norm}[1]{{\left| #1 \right|}}
\newcommand{\floor}[1]{{\left\lfloor #1 \right\rfloor}}

\newcommand{\brm}[1]{\operatorname{#1}}
\newcommand{\tw}[0]{\brm{tw}}
\newcommand{\bn}[0]{\brm{bn}}
\newcommand{\sn}[0]{\brm{sn}}

\newcommand{\dist}[0]{\brm{dist}}
\newcommand{\depth}[0]{\brm{depth}}
\newcommand{\org}[0]{\brm{org}}
\newcommand{\dual}[0]{\brm{dual}}
\newcommand{\boundary}[0]{\brm{boundary}}

\begin{document}

\maketitle

\begin{abstract}
Treewidth is an important structural graph parameter that quantifies how closely a graph resembles a tree-like structure. It has applications in many algorithmic and combinatorial problems. In this paper, we study the treewidth of \emph{outer $k$-planar} graphs, that is, graphs admitting a \emph{convex drawing} (a straight-line drawing where all vertices lie on a circle) in which every edge crosses at most $k$ other edges. We also consider the more general class of \emph{outer min-$k$-planar} graphs, which are graphs admitting a convex drawing where for every crossing of two edges at least one of these edges is crossed at most $k$ times.

Firman, Gutowski, Kryven, Okada and Wolff~[GD 2024] proved that every outer $k$-planar graph has treewidth at most $1.5k+2$ and provided a lower bound of $k+2$ for even $k$. We establish a lower bound of $1.5k+0.5$ for every odd $k$. Additionally, they showed that every outer min-$k$-planar graph has treewidth at most $3k+1$. We improve this upper bound to $3 \cdot \floor{k/2}+4$. 

Our approach also allows us to upper bound the \emph{separation number}, a parameter closely related to treewidth, of outer min-$k$-planar graphs by $2 \cdot \floor{k/2}+4$. This improves upon the previous bound of $2k+1$ and achieves a bound with an optimal multiplicative constant.
\end{abstract}

\section{Introduction}

In this paper, we study classes of graphs admitting a \emph{convex drawing} with a bounded number of edge crossings. A convex drawing is a straight-line drawing with all vertices drawn on a common circle. Bannister and Eppstein~\cite{be-cm1p2p-GD14, be-cm1p2pdgbtw-JGAA18} proved that the treewidth of graphs admitting a convex drawing with at most $k$ crossings in total is bounded by a linear function of $\sqrt{k}$. For fixed $k$, they also provided a linear-time algorithm deciding whether a given graph admits such a drawing (using Courcelle’s theorem~\cite{courcelle1990}). Another well-studied class of graphs in this context is the class of \emph{outer $k$-planar} graphs, that is, graphs that admit a convex drawing in which every edge crosses at most $k$ other edges. These graphs have treewidth bounded by a linear function of $k$, which was first proven by Wood and Telle~\cite[Proposition~8.5]{wt-pdcng-NYJM07}. The authors of~\cite{BeyondOuterplanarity}, also using Courcelle’s theorem, presented, for any fixed $k$, a linear-time algorithm that tests whether a given graph is maximal outer $k$-planar.
Recently, Kobayashi, Okada and Wolff~\cite{kobayashi2025}, for any fixed $k$, provided a polynomial-time algorithm to test whether a given graph is outer $k$-planar and proved that recognising outer $k$-planar graphs is {\sf XNLP}-hard.

For disambiguation, we recall the definition of \emph{$k$-outerplanar} graphs. A graph is \emph{outerplanar} if it has a planar drawing with all vertices lying on the outer face. A graph is \emph{$1$-outerplanar} when it is outerplanar. A graph is \emph{$k$-outerplanar} for $k > 1$ when it has a planar drawing such that after removing the vertices of the outer face, each of the remaining components is $(k-1)$-outerplanar.

We mainly study the treewidth of outer $k$-planar graphs and \emph{outer min-$k$-planar} graphs. A~graph is outer min-$k$-planar if it admits a convex drawing in which, for every crossing of two edges, at least one of these edges is crossed at most $k$ times. Firman, Gutowski, Kryven, Okada and Wolff~\cite{FirmanGKOW24} proved that outer $k$-planar graphs have treewidth at most $1.5k+2$ and outer min-$k$-planar graphs have treewidth at most $3k+1$. To obtain these results, they showed that every outer $k$-planar graph admits a \emph{triangulation} of the outer cycle such that every edge of the triangulation is crossed at most $k$ times by the edges of the graph. A~similar property was proven for outer min-$k$-planar graphs.

Another property closely related to treewidth is the separation number of a graph. A~\emph{separation} of a graph $G$ is a pair $(A, B)$ of subsets of $V(G)$ such that $A \cup B = V(G)$ and there are no edges between the sets $A \setminus B$ and $B \setminus A$. The \emph{order} of a separation is $\norm{A \cap B}$. A~separation is \emph{balanced} if $\norm{A\setminus B} \leq \frac{2}{3} \norm{V(G)}$ and $\norm{B\setminus A} \leq \frac{2}{3} \norm{V(G)}$. The \emph{separation number} of a graph $G$, denoted $\sn(G)$, is the minimum integer $a$ such that every subgraph of $G$ has a balanced separation of order at most $a$. Robertson and Seymour~\cite{ROBERTSON1986309} proved that $\sn(G) \leq \tw(G)+1$ for every graph~$G$. From the other side, Dvo\v{r}\'ak and Norin~\cite{dn-tgbs-JCTB19} showed that $\tw(G) \le 15\sn(G)$. Recently,  Houdrouge, Miraftab and Morin~\cite{houdrouge2025} provided a more constructive proof of an analogous inequality, but with a worse multiplicative constant.

\subparagraph*{Our contribution.}
The authors of~\cite{FirmanGKOW24} proved that every outer $k$-planar graph has treewidth at most $1.5k+2$. They also presented a lower bound of $k+2$ for every even $k$. We present an infinite family of outer $k$-planar graphs with treewidth at least $1.5k+0.5$, showing that the multiplicative constant $1.5$ in the upper bound cannot be improved; see \cref{sec:lower}.

We also improve the upper bounds for the treewidth and separation number of outer min-$k$-planar graphs. 
It was previously known that the treewidth of such graphs is at most $3k+1$ and the separation number is at most $2k+1$~\cite{FirmanGKOW24}.
We give an upper bound of $3 \cdot \floor{k/2} + 4$ for the treewidth (see \cref{sec:upper}) and an upper bound of $2 \cdot \floor{k/2} + 4$ for the separation number (see \cref{sec:separators}). Both multiplicative constants are optimal, as the lower bounds for outer $k$-planar graphs also hold for outer min-$k$-planar graphs -- namely, our lower bound of $1.5k+0.5$ for the treewidth and the lower bound of $k+2$ for the separation number presented in~\cite{FirmanGKOW24}.

\subparagraph*{Related results.}
A similar type of result is known for the pathwidth of \emph{2-layer $k$-planar} graphs. A 2-layer $k$-planar graph is a bipartite graph that admits a straight-line drawing where all vertices lie on two parallel lines and every edge crosses at most $k$ other edges.
Angelini, Da Lozzo, Förster, and Schneck \cite{2LayerkPlanar} showed that every 2-layer $k$-planar graph has pathwidth at most $k+1$. Recently, Okada \cite{okada2025} proved that this bound is sharp by constructing a 2-layer $k$-planar graph with pathwidth $k+1$ for every $k \geq 0$.

\section{Preliminaries}

Let $G$ be a graph. By $V(G)$ and $E(G)$ we denote the set of vertices and edges of $G$, respectively. For an edge that connects vertices $u$ and $v$, we use the compact notation $uv$, instead of $\set{u, v}$. For a directed edge, we use the standard notation $(u, v)$. Let $\deg(v)$ denote the degree of a vertex $v$, and let $\Delta(G)$ denote the maximum degree of a vertex of $G$.

For a graph $G$, a subgraph \emph{induced} by a set $U \subseteq V(G)$, denoted $G[U]$, is a subgraph with vertex set $U$ and all edges of $G$ between the vertices of $U$. A \emph{spanning tree} of a graph $G$ is a subgraph of $G$ containing all the vertices of $G$ that is a tree. By $\dist_G(v, w)$, we denote the distance (i.e. the length of the shortest path) between $v$ and $w$ in a graph $G$. For any tree $T$ rooted at vertex $r$, we define the depth of a vertex $v$ as $\depth_T(v) = \dist_T(r, v)$. We may omit subscripts if they are clear from the context.

A \emph{tree decomposition} $\mathcal{T} = (T, B)$ of a graph $G$ is a collection of \emph{bags}, $\{B_x: x \in V(T)\}$, indexed by the vertices of a tree $T$. The bags are subsets of $V(G)$ and satisfy the following properties:
\begin{enumerate}
    \item for every vertex $v \in V(G)$, the set $\{x: v \in B_x\}$ induces a non-empty subtree of $T$;
    \item for every edge $uv \in E(G)$, there exists a bag containing both $u$ and $v$.
\end{enumerate}
The \emph{width} of a given tree decomposition is the size of the largest bag minus one. The \emph{treewidth} of a graph $G$, denoted by $\tw(G)$, is the minimum width of any tree decomposition of $G$.

A set $\mathcal{B}$ of non-empty subsets of $V(G)$ is a \emph{bramble} if:
\begin{enumerate}
    \item for every $X \in \mathcal{B}$, the induced subgraph $G[X]$ is connected;
    \item for every $X_1, X_2 \in \mathcal{B}$, the induced subgraph $G[X_1 \cup X_2]$ is connected. In other words, each $X_1, X_2 \in \mathcal{B}$ either share a common vertex or there exists an edge of $G$ incident to both $X_1$ and~$X_2$. 
\end{enumerate}
A \emph{hitting set} of a bramble is a set of vertices with non-empty intersection with every element of $\mathcal{B}$. The \emph{order} of a bramble is the size of its smallest hitting set. The \emph{bramble number} of a graph $G$, denoted by $\bn(G)$, is the maximum order of any bramble of $G$.

The following result by Seymour and Thomas shows the relation between the bramble number and treewidth.
\begin{theorem}[Seymour and Thomas, \cite{seymour1993graph}] \label{thm:tw_bramble}
For every graph $G$, $\tw(G) = \bn(G)-1$.
\end{theorem}

We say that a graph $G$ is a \emph{minor} of a graph $H$ if $G$ can be obtained from $H$ by a sequence of vertex deletions, edge deletions or edge contractions. The contraction of an edge $uv$ is an operation that replaces vertices $u$ and $v$ with a new vertex adjacent to every vertex other than $u$ and $v$ that was adjacent to $u$ or $v$. It is a well-known fact that if $G$ is a minor of $H$, then $\tw(G) \leq \tw(H)$. A proof of this fact can be found in \cite{Bodlaender1998}.

In the remainder of this section, we introduce some notation and simple observations regarding drawings. A \emph{convex drawing} of a graph $G$ is a straight-line drawing where the vertices of $G$ are placed on distinct points of a circle. Given a cyclic order $(v_1, \ldots, v_n)$ of vertices, we say that an edge $v_iv_j$ with $i < j$ \emph{crosses} an edge $v_{i'}v_{j'}$ with $i' < j'$ if either $1 \leq i < i' < j < j' \leq n$ or $1 \leq i' < i < j' < j \leq n$. We only consider convex drawings where no three pairwise non-adjacent edges pass through the same point. An \emph{outer $k$-planar drawing} of a graph is a convex drawing such that every edge crosses at most $k$ other edges. An \emph{outer min-$k$-planar drawing} of a graph is a convex drawing such that for every crossing of two edges, at least one of these edges crosses at most $k$ other edges.

An outer min-$k$-planar graph $G$ is \emph{maximal outer min-$k$-planar} if for every $\{u, v\} \subseteq V(G)$ with $uv \not\in E(G)$, the graph $G+uv$ is not outer min-$k$-planar.

\begin{observation} \label{obs:max_okpl}
    Let $G$ be a maximal outer min-$k$-planar graph with at least three vertices. Then, in every outer min-$k$-planar drawing of $G$, the outer face is bounded by a cycle.
\end{observation}

\begin{proof}
    Consider an outer min-$k$-planar drawing $\Gamma$ of $G$. Let $u$ and $v$ be consecutive vertices in the cyclic order defined by $\Gamma$. Suppose, for contradiction, that $uv \not\in E(G)$. Note that the graph $G+uv$ has an outer min-$k$-planar drawing defined by the same cyclic order as $\Gamma$, contradicting the maximality of $G$.
\end{proof}

A graph $G$ is \emph{expanded outer min-$k$-planar} if $G$ is an outer min-$k$-planar graph with $\Delta(G) \leq 3$ and its outer face is bounded by a cycle in some outer min-$k$-planar drawing of~$G$.

\begin{observation} \label{obs:expand_okpl}
    Every outer min-$k$-planar graph $G$ is a minor of an expanded outer min-$k$-planar graph $G'$.
\end{observation}

\begin{proof}
     Let us assume that $G$ is maximal outer min-$k$-planar. Now, in order to obtain $G'$ from $G$, we perform the following transformation to every vertex $v$ of $G$ with $\deg(v) \geq 4$. The transformation is depicted in \cref{fig:expansion}. Let $w_0, w_1, \ldots, w_s, w_{s+1}$ be all neighbors of $v$ in clockwise order, with edges $vw_0$ and $vw_{s+1}$ incident to the outer face of $G$. We replace $v$ with a path $v_1, \ldots, v_s$, put it on the outer face of $G$ in counterclockwise order, in the place of~$v$. We connect this path to vertices $w_0$ and $w_{s+1}$ by adding edges $v_1w_0$ and $v_sw_{s+1}$. Finally, for every $1 \leq i \leq s$, we add an edge $v_iw_i$ that corresponds to an edge $vw_i$ in the original graph. It is easy to see that $G$ is a minor of $G'$ and the ordering of corresponding edges in $G'$ matches the one in $G$. Moreover, the crossings in the resulting graph naturally correspond to the crossings in the original graph.
\end{proof}

\begin{figure}
    \centering
    \includegraphics[width=0.9\linewidth]{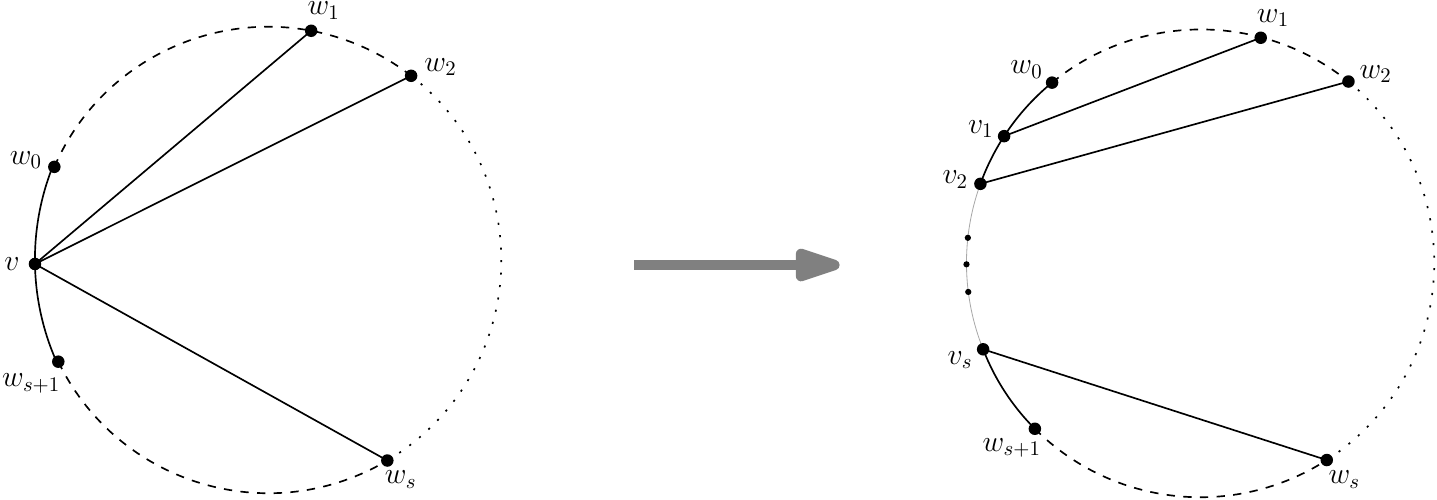}
    \caption{The transformation described in \cref{obs:expand_okpl}.}
    \label{fig:expansion}
\end{figure}

The vertices $v_1, \ldots, v_s$ obtained in the proof as a replacement of $v$ are called \emph{images} of $v$. The vertex $v$ is the \emph{origin} of these vertices, which we denote as $\org(v_i) = v$. If the transformation was not performed for some vertex $v$ of $G$, i.e.~$\deg(v) \leq 3$, then $v$ is an image and origin of itself.

Since removing edges increases neither the treewidth nor the separation number, we are interested in the properties of maximal outer min-$k$-planar graphs. Also, taking a minor does not increase treewidth, so we work with expanded graphs when establishing upper bounds on treewidth.

\section{Lower bound on the treewidth of outer $k$-planar graphs}
\label{sec:lower}

In this section, we construct an infinite family of outer $k$-planar graphs with treewidth at least $1.5k+0.5$. This improves the previous lower bound of $k+2$ that was presented in~\cite{FirmanGKOW24}. We begin by defining the necessary graphs. 

For positive integers $m$ and $n$, let $X_{m,n}$ denote the grid of $m$ rows and $n$ columns, i.e.~a graph with \[V(X_{m,n}) = \set{x_{i,j}: 1\le i \le m, 1\le j \le n} \text{ and } E(X_{m, n}) = \set{x_{i,j}x_{k,l}: \norm{i-k} + \norm{j-l} = 1}.\]
For a positive integer $k$, let $Q_k$ be a copy of the grid $X_{2k,2k}$ and let $R_k$ be a copy of $X_{2k(k+1),k}$. Denote by $v_{i,j}$, for $1\le i,j \le 2k$, the vertex in the $i$-th row and $j$-th column of $Q_k$, and by $u_{i,j}$, for $1 \le i \le 2k(k+1)$, $1 \le j \le k$, the vertex in the $i$-th row and $j$-th column of $R_k$. Let $G_k$ be a graph such that $V(G_k) = V(Q_k) \cup V(R_k)$ and 
\[E(G_k) = E(Q_k) \cup E(R_k) \cup \set{v_{i, 2k}u_{(i-1)(k+1)+j,1}: 1 \leq i \leq 2k, 1 \leq j \leq k+1};\]
see \cref{fig:graph_Gk}. 
For $1\le i \le 2k(k+1)$, let the $i$-th \emph{extended row} of $G_k$ be the union of the $i$-th row of $R_k$ and the $\lceil\frac{i}{k+1}\rceil$-th row of $Q_k$. Note that each row of $Q_k$ is contained in $k+1$ extended rows and the graph induced by each extended row is a path.

The graph $G_k$ was previously introduced by Kammer and Tholey~\cite{kammer2009lower} as an example of tightness of the upper bound on the treewidth of $k$-outerplanar graphs. They used the \emph{cops and robber game} to establish a lower bound on the treewidth of $G_k$. Below, we present a proof using brambles.

\begin{figure}
    \centering
    \includegraphics[width=0.7\linewidth]{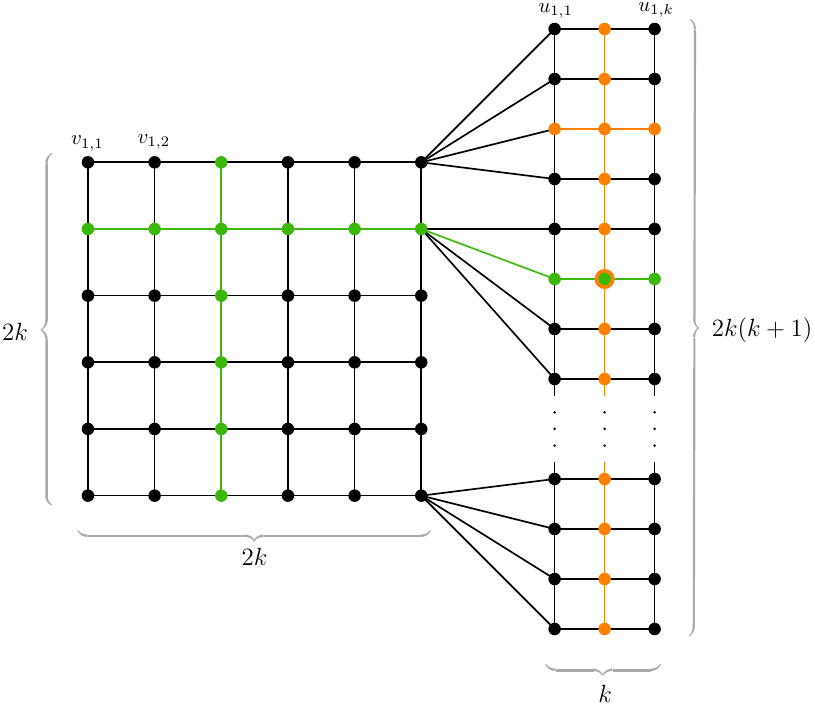}
    \caption{The graph $G_k$, for $k=3$, with a subgraph of $\mathcal{B}_1$ colored green and a subgraph of $\mathcal{B}_2$ colored orange.}
    \label{fig:graph_Gk}
\end{figure}

\begin{theorem}[Kammer and Tholey, \cite{kammer2009lower}] \label{thm:Hk_tw}
    For every $k \geq 1$, $\tw(G_k) = 3k-1$.
\end{theorem}

\begin{proof}
    Notice that the drawing of $G_k$ in \cref{fig:graph_Gk} is $k$-outerplanar. By the fact that $k$-outerplanar graphs have treewidth at most $3k-1$ \cite[Theorem 83]{Bodlaender1998}, we get $\tw(G_k) \leq 3k-1$.

    To prove that $\tw(G_k) \geq 3k-1$, we construct a bramble of order $3k$. Then, using \cref{thm:tw_bramble}, we get $\tw(G_k) \geq 3k-1$. 
    Let $\mathcal{B}_1$ be a family consisting of every subset of $V(G_k)$ that is a union of an extended row of $G_k$ and a column of $Q_k$. Let $\mathcal{B}_2$ be a family consisting of every subset of $V(G_k)$ that is a union of a row of $R_k$ and a column of $R_k$. The set $\mathcal{B} = \mathcal{B}_1 \cup \mathcal{B}_2$ forms a bramble of $G_k$, since each subgraph induced by an element of $\mathcal{B}$ is connected and every two such subgraphs have at least one common vertex.

    Consider any hitting set $S$ of $\mathcal{B}$. Let $q$ and $r$ be the number of vertices of $S$ in $V(Q_k)$ and in $V(R_k)$, respectively. We would like to show that $|S| = q+r \geq 3k$. Note that $r \geq k$, as otherwise there is a row and a column of $R_k$ not containing any element of $S$, and thus there is an element of $\mathcal{B}_2$ not hit by $S$.

    If $q \geq 2k$, then $q+r \geq 3k$. Otherwise, let $q = 2k-l$ for some positive integer $l$. Now, we can find at least $l$ columns and at least $l$ rows of $Q_k$ not intersecting $S$. These $l$ rows are contained in $l(k+1)$ extended rows. Each of them has to intersect $S$ at some vertex of $R_k$, because otherwise we can find a column of $Q_k$ and an extended row not intersecting $S$ that form an element of $\mathcal{B}_1$. The extended rows restricted to $R_k$ are pairwise disjoint, so we have $r \geq l(k+1)$. Summing up, we get $q+r \geq 2k-l + l(k+1) = 2k + lk \geq 2k+k = 3k$, which concludes the proof.
\end{proof}

Let $F_k$ be the following modification of $G_k$ depicted in \cref{fig:graph_Fk}.
We set $\ell(i) = (k-i)(k+1)$ for $1 \le i \le k$, and $\ell(i)=(i-k-1)(k+1)$ for $k+1 \le i \le 2k$.
We remove every edge between the grids $Q_k$ and $R_k$.
For every $1 \leq i \leq 2k$, we add a path $Z_i$ of length $\ell(i)$ on new vertices $z_{i,0},z_{i,1},\ldots,z_{i,\ell(i)}$.
Note that
\[k^2-1 = \norm{V(Z_1)}=\norm{V(Z_{2k})} > \norm{V(Z_2)}=\norm{V(Z_{2k-1})} > \ldots > \norm{V(Z_k)}=\norm{V(Z_{k+1})}=1.\]
We also add a path $W_i$ of length $k$ on new vertices $w_{i,1},w_{i,2},\ldots,w_{i,k+1}$.
We connect $v_{i,2k}$ with $Z_i$ by adding the edge $v_{i,2k}z_{i,0}$. Next, we connect $Z_i$ with $W_i$ by adding the edge $z_{i,\ell(i)}w_{i,k+1}$ for $1 \leq i \leq k$, or the edge $z_{i,\ell(i)}w_{i,1}$ for $k+1 \leq i \leq 2k$. Finally, we connect $W_i$ with $R_k$ by adding the edges $w_{i,j}u_{(i-1)(k+1)+j,1}$ for every $1 \le j \le k+1$.

To see that $G_k$ is a minor of $F_k$, it is enough to contract, for every $1 \le i \le 2k$, vertex $v_{i,2k}$ with all vertices of the paths $Z_i$ and $W_i$. Since taking a minor does not increase the treewidth, we obtain the following corollary.

\begin{figure}
    \centering
    \includegraphics[width=1\linewidth]{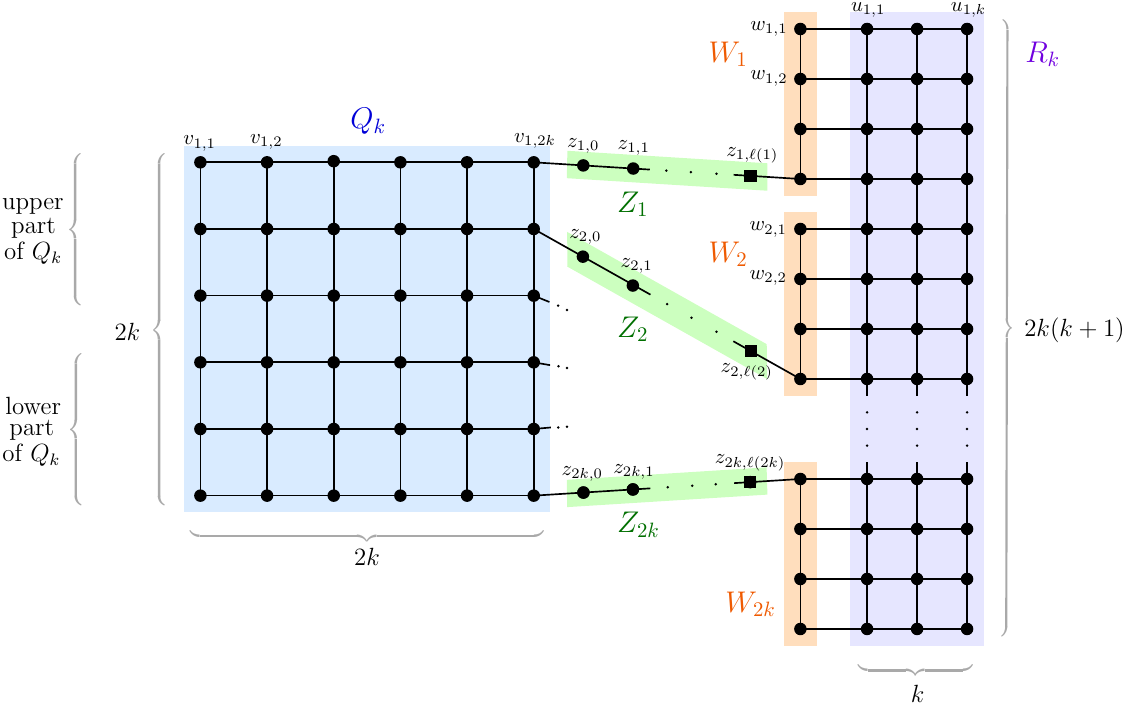}
    \caption{The graph $F_k$, which is a modification of the graph $G_k$, for $k=3$.}
    \label{fig:graph_Fk}
\end{figure}

\begin{corollary} \label{cor:Fk_tw}
    For every $k \geq 1$, $\tw(F_k) \geq 3k-1$.
\end{corollary}

\begin{theorem} \label{thm:Fk_drawing}
    For every $k \geq 1$, The graph $F_k$ has an outer $(2k-1)$-planar drawing.
\end{theorem}

\begin{proof}
    We describe an outer $(2k-1)$-planar drawing of $F_k$, as depicted in \cref{fig:drawing_Fk}. We call the set of vertices $\{v_{i,j}: 1 \leq i \leq k, 1 \leq j \leq 2k\}$ the \emph{upper part} of $Q_k$. The other vertices of $Q_k$ are called the \emph{lower part} of $Q_k$. We define a cyclic order of the vertices of $F_k$ by arranging them in clockwise order from some selected starting point on a circle. 

    \begin{figure}
        \centering

        \begin{subfigure}[b]{1\textwidth}
        \centering
           \includegraphics[width=0.8\linewidth]{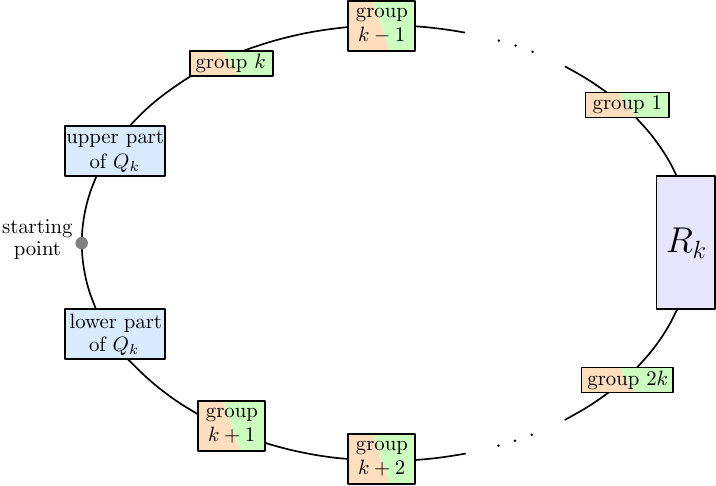}
            \subcaption{Overview of the outer $(2k-1)$-planar drawing of $F_k$. \vskip 6mm}
            \label{fig:Fk_overview}
        \end{subfigure}

        \begin{subfigure}[b]{1\textwidth}
           \includegraphics[width=0.9\linewidth]{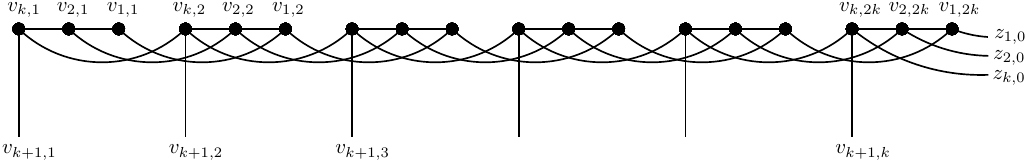}
           \subcaption{Drawing of the upper part of the grid $Q_k$. \vskip 6mm}
           \label{fig:drawing_Qk} 
        \end{subfigure}

        \begin{subfigure}[b]{1\textwidth}
           \includegraphics[width=1\linewidth]{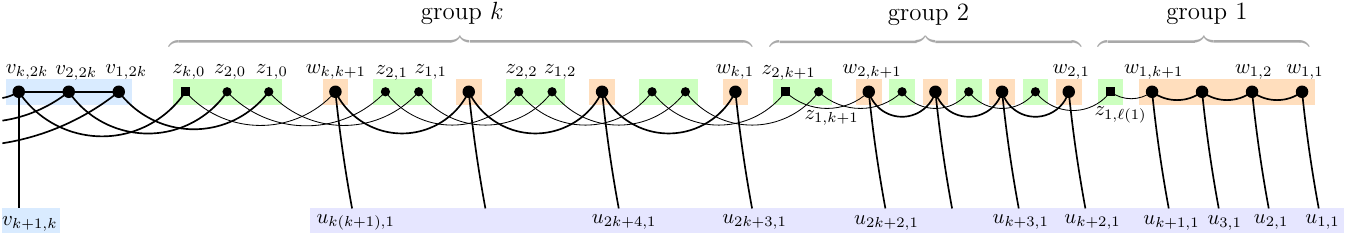}
           \subcaption{Drawing of the groups $k, k-1, \ldots, 1$ connected to the upper part of $Q_k$. \vskip 6mm }
           \label{fig:drawing_groups} 
        \end{subfigure}

        \begin{subfigure}[b]{1\textwidth}
           \includegraphics[width=1\linewidth]{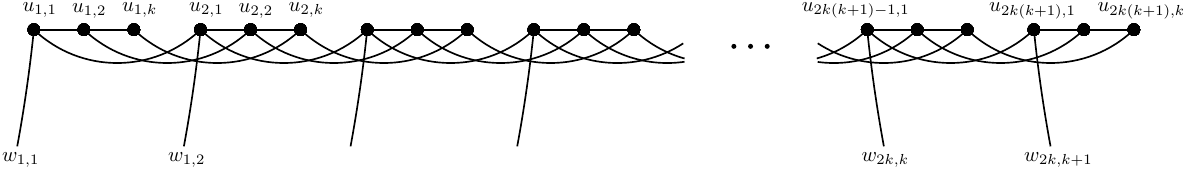}
           \subcaption{Drawing of the grid $R_k$.}
           \label{fig:drawing_Rk} 
        \end{subfigure}
        
        \caption{Fragments of the outer $(2k-1)$-planar drawing of $F_k$, where $k=3$. The drawings are not shown as straight-line drawings, but illustrate the order in which the vertices are placed. A straight-line drawing can be easily obtained by placing vertices in this order on the circle.}
        \label{fig:drawing_Fk}
    \end{figure}
    
    First, we place the vertices of the upper part of $Q_k$ in column-by-column order (see \cref{fig:drawing_Qk}):
    \[ v_{k, 1}, \ldots, v_{2, 1}, v_{1, 1}, \quad v_{k, 2}, \ldots, v_{2, 2}, v_{1, 2}, \quad \ldots \quad v_{k, 2k}, \ldots, v_{2, 2k}, v_{1, 2k}.\]

    We divide the vertices of paths $Z_i$ and $W_i$, for $1 \le i \le k$, into $k$ groups, as follows.
    The $i$-th group contains the vertices of $W_i$ and the vertex $z_{i, \ell(i)}$. If $i \geq 2$, the $i$-th group also contains vertices $z_{a,b}$ for $1 \leq a < i$ and $(k-i)(k+1) \leq b < (k-i+1)(k+1)$. In the drawing, we place the groups of indices $k,k-1,\ldots,1$, in this order. We arrange the vertices in the $i$-th group, for $2 \leq i \leq k$, in the order (see \cref{fig:drawing_groups}):
    \begin{align*}
        & z_{i, \ell(i)}, \\
        & z_{i-1, (k-i)(k+1)}, z_{i-2, (k-i)(k+1)},\ldots, z_{1, (k-i)(k+1)}, w_{i,k+1}, \\
        & z_{i-1, (k-i)(k+1)+1}, z_{i-2, (k-i)(k+1)+1},\ldots, z_{1, (k-i)(k+1)+1}, w_{i,k}, \\ 
        & \vdots \\
        & z_{i-1, (k-i+1)(k+1)-1}, z_{i-2, (k-i+1)(k+1)-1},\ldots, z_{1, (k-i+1)(k+1)-1}, w_{i,1}.
    \end{align*}
    The group of index 1 has vertices arranged in the order: $z_{1, \ell(1)}, w_{1,k+1}, w_{1,k}, \ldots, w_{1,1}$.

    Next, we put the vertices of $R_k$ in row-by-row order (see \cref{fig:drawing_Rk}):
    \[u_{1, 1}, u_{1, 2}, \ldots, u_{1, k}, \quad u_{2, 1}, u_{2, 2}, \ldots, u_{2, k}, \quad \ldots \quad u_{2k(k+1), 1}, u_{2k(k+1), 2}, \ldots, u_{2k(k+1), k}.\]

    The vertices of $F_k$ that are not placed yet are in the lower part of $Q_k$ or in the paths $Z_i$, $W_i$ with $k+1 \le i \le 2k$. We arrange them in counterclockwise order from the starting point and place them between the starting point and the vertices of $R_k$. The order is symmetric, with respect to the starting point, to the one used to arrange the upper part of $Q_k$ and the paths $Z_i$, $W_i$ with $1 \le i \le k$.
    Every vertex $v_{i,j}$, where $k+1 \leq i \leq 2k$ and $1 \leq j \leq 2k$, is placed symmetrically to $v_{2k-i+1, j}$. Vertex $z_{i,j}$, where $k+1 \leq i \leq 2k$ and $0 \leq j \leq \ell(i)$, is placed symmetrically to $z_{2k-i+1, j}$, and vertex $w_{i,j}$, where $k+1 \leq i \leq 2k$ and $1 \leq j \leq k+1$, is placed symmetrically to $w_{2k-i+1, k-j+2}$. The symmetrical drawing of the $i$-th group, for every $1 \leq i \leq k$, forms the group of index $2k-i+1$.

    Now, we partition the edges into several numbered types. For the edges of each type, we show that they cross at most $2k-1$ other edges.
        \vspace{1mm} \newline 
        Edges of $Q_k$:
        \begin{enumerate}
            \item The ``column'' edges of the upper or lower part of $Q_k$, that is, the edges $v_{i,j}v_{i+1,j}$, for $1 \leq i \leq 2k-1, i \neq k$ and $1 \leq j \leq 2k$. They cross no other edges.
            \item The ``column'' edges between the upper part and the lower part of $Q_k$, that is, the edges $v_{k,j}v_{k+1,j}$, for $1 \leq j \leq 2k$. Each of these edges crosses $k-1$ edges of type 3 of the upper part of $Q_k$, and $k-1$ edges of type 3 of the lower part of $Q_k$. The edge $v_{k,1}v_{k+1,1}$ does not cross any edges.   
            \item The ``row'' edges of $Q_k$, that is, the edges $v_{i, j}v_{i,j+1}$, for $1 \leq i \leq 2k$ and $1 \leq j \leq 2k-1$.
            Each of these edges crosses at most $2(k-1)$ edges of type 3 (and type 4, for $j=2k-1$), and additionally at most one edge of type~2.
        \end{enumerate}
        \vspace{1mm} Edges between $Q_k$ and the groups:
        \begin{enumerate}
            \setcounter{enumi}{3}
            \item Each edge $v_{i,2k}z_{i,0}$, for $1 \leq i \leq k$, crosses $2(k-1)$ edges in total: $2(i-1)$ edges of types 3,~4 incident to vertices $v_{j,2k}$, for all $j \in \set{1, \ldots, i-1}$; and $2(k-i)$ edges incident to vertices $z_{j,0}$, for all $j \in \set{i+1, \ldots, k}$. By symmetry, each edge $v_{i,2k}z_{i,0}$, for $k+1 \leq i \leq 2k$, also crosses exactly $2(k-1)$ edges.
        \end{enumerate}
        \vspace{1mm} Edges of the groups:
        \begin{enumerate}
            \setcounter{enumi}{4}
             \item Each edge $z_{i,\ell(i)}w_{i, k+1}$, for $1 \leq i \leq k$, crosses exactly $2(i-1)$ edges incident to the vertices $z_{i-1, \ell(i)}, \ldots, z_{1,\ell(i)}$. By symmetry, each edge $z_{i,\ell(i)}w_{i,1}$, for $k+1 \leq i \leq 2k$, also crosses exactly $2(2k-i)$ edges.
            \item Each edge of the path $Z_i$, for $1 \leq i \leq 2k$, crosses at most $2(k-2)$ edges incident to at most $k-2$ vertices of some paths $Z_j$ with $j \in \set{1, \ldots, 2k}$, and at most three edges incident to some vertex $w_{a,b} \in V(W_j)$ with $j \in \set{1, \ldots, 2k}$.
            \item Each edge of the path $W_i$, for $1 \leq i \leq k$, crosses $2(i-1)$ edges incident to $i-1$ vertices of some paths $Z_j$ with $j \in \set{1, \ldots, k}$. By symmetry, each edge of the path $W_i$, for $k+1 \leq i \leq 2k$, also crosses $2(2k-i)$ edges.
        \end{enumerate}        
        \vspace{1mm} Edges between the groups and $R_k$:
        \begin{enumerate}
            \setcounter{enumi}{7}
            \item Each edge $w_{i,j} u_{(i-1)(k+1)+j,1}$, for $1 \leq i \leq 2k$ and $1 \leq j \leq k+1$, crosses at most $k-1$ edges of some paths $Z_a$ with $a \in \set{1, \ldots, 2k}$, and at most $k-1$ edges of $R_k$ of type 10.
        \end{enumerate}
        \vspace{1mm} Edges of $R_k$:
        \begin{enumerate}
            \setcounter{enumi}{8}
            \item The ``row'' edges of $R_k$, that is, the edges $u_{i,j}u_{i,j+1}$, for $1 \leq i \leq  2k(k+1)$ and $1 \leq j \leq k-1$. They cross no other edges. 
            \item The ``column'' edges of $R_k$, that is, the edges $u_{i,j}u_{i+1,j}$, for $1 \leq i \leq 2k(k+1)-1$ and $1 \leq j \leq k$. Each of these edges crosses at most $2(k-1)$ other edges of this type and at most one edge of type~8.
        \end{enumerate}
\end{proof}

\begin{theorem}
    For every odd positive integer $k$, there exists an outer $k$-planar graph $G$ with $\tw(G) \geq 1.5k+0.5$.
\end{theorem}

\begin{proof}
    By \cref{thm:Fk_drawing}, the graph $F_{\frac{k+1}{2}}$ is outer $k$-planar, and by \cref{cor:Fk_tw}, it has treewidth at least $3\cdot \frac{k+1}{2}-1 = 1.5k+0.5$.
\end{proof}

\section{Upper bound on the treewidth of outer min-$k$-planar graphs}
\label{sec:upper}

In this section, we establish an upper bound on the treewidth of outer min-$k$-planar graphs. We improve the previous bound of $3k+1$ presented in~\cite{FirmanGKOW24} to $3\cdot \floor{k/2}+4$. We begin by introducing required notation.

For an outer min-$k$-planar graph $G$ with a given drawing $\Gamma$, we define the \emph{planarization} $G_P$ (with respect to $\Gamma$) as the graph whose vertex set is the union of $V(G)$ and all crossing points of edges of $G$. We say that a vertex $w \in V(G_P)$ \emph{lies} on an edge $uv \in E(G)$ if $w$ is an endpoint of $uv$, or if the crossing point corresponding to $w$ belongs to the segment representing the drawing of $uv$ in $\Gamma$. Graph $G_P$ contains an edge between two vertices if and only if they are consecutive vertices lying on the drawing of some edge of $G$. Observe that $G_P$ is a planar graph. We say that an edge $xy \in E(G_P)$ \emph{lies} on an edge $uv \in E(G)$ if both $x$ and $y$ lie on $uv$ in $\Gamma$. Furthermore, we say that a vertex $v \in V(G_P)$ is \emph{outer} if it is incident to the outer face of $G_P$. Otherwise, $v$ is an \emph{inner} vertex. As we consider only maximal outer min-$k$-planar graphs $G$, the outer vertices of $G_P$ are exactly the vertices of $G$, while the inner vertices of $G_P$ are exactly the crossing points of edges of $G$. By $f_o$ we denote the outer face of $G_P$.

For a planar graph $G$, let $G^*$ denote its dual graph. Let $f^* \in V(G^*)$ denote the vertex dual to the face $f$ of $G$, and let $e^* \in E(G^*)$ denote the edge dual to the edge $e \in E(G)$. We remark that $G^*$ can be drawn on the drawing of $G$ in a way that $f^*$ is on the face $f$ and the drawing of $e^*$ is a curve that crosses the edge $e$ exactly once and passes only through the faces corresponding to the endpoints of $e^*$. 

The following lemma shows a bijection between a spanning tree $T$ of a planar graph $G$ and a spanning tree $T^*$ of $G^*$, where $T^* = \dual(T)$. We also use the notation $\dual(T^*)$ to denote $T$.

\begin{lemma} [Folklore] \label{lem:spanning_dual}
    Let $T$ be a spanning tree of a planar graph $G$. Then $T^*$ with $V(T^*) = V(G^*)$ and $E(T^*) = \set{e^*:e\in E(G) \setminus E(T)}$ is a spanning tree of $G^*$.
\end{lemma}

The next lemma proves that there exists a spanning tree preserving shortest paths from a given vertex. Such a tree can be found using a breadth-first search.

\begin{lemma} [Folklore] \label{lem:bfs_tree}
    Let $G$ be a graph and let $r$ be a vertex of $G$. Then there exists a spanning tree $T$ of $G$ rooted at $r$ such that $\depth_T(v) = \dist_G(r, v)$ for every vertex $v$ of $G$.
\end{lemma}

\begin{lemma} \label{lem:outerkpl_dual_dist}
    Let $G$ be an expanded outer min-$k$-planar graph with its planarization $G_P$. Then $\dist(f^*, f_o^*) \leq \floor{k/2}+1$ for every vertex $f^* \in V(G_P^*)$. 
\end{lemma}

\begin{proof}
    Let $f$ be an inner face of $G_P$. If $f$ is adjacent to $f_o$, then $\dist(f^*, f^*_o) = 1$. Otherwise, let $v$ be a vertex of $G_P$ incident to $f$. Since $G$ is expanded, the vertex $v$ is inner, and hence it lies on an edge $e$ of $G$ that crosses at most $k$ other edges. Let $v_0, v_1, \ldots, v_s, v_{s+1}, \ldots, v_{s+t+1}$ be all vertices lying on $e$, listed in the order along $e$, where $v_s = v$ and $v_{s+1}$ is a neighbor of $v$ that is incident to $f$. We may assume that $s\leq t$, i.e.~$v_s$ is closer to an endpoint of the edge $e$ than $v_{s+1}$ to the other endpoint of $e$. Note that at most $k+2$ vertices lie on $e$ (two endpoints and at most $k$ crossing points), so $s+t+2 \leq k+2$. Together with the previous inequality, this implies $s \leq k/2$. The number $s$ is an integer so $s \leq \floor{k/2}$.

     We inductively define a sequence $w_s, w_{s-1}, \ldots, w_0$ of vertices. Vertex $w_s$ is the neighbor of $v_s$ that is incident to $f$ and does not lie on $e$. For every $i = s-1,\ldots,0$, the vertex $w_i$ is one of the two neighbors of $v_i$ not lying on $e$, chosen so that $w_{i}$, $v_{i}$, $v_{i-1}$ and $w_{i-1}$ are incident to the same face of $G_P$. Let $e_i$ denote the edge $v_iw_i$. Note that the path formed by the edges $e_s^*, \ldots, e_0^*$ connects $f^*$ with $f_o^*$ (since $e_0$ is incident to $f_o$). Hence, $\dist(f^*, f_o^*) \leq s+1 \leq \floor{k/2}+1$.
\end{proof}

Let $v$ be an inner vertex of $G_P$. Since we forbid common crossing points of three edges of $G$, the vertex $v$ has four neighbors, which we denote by $w_1, w_2, w_3, w_4$ in clockwise order. The \emph{splitting} of the vertex $v$ replaces it with two vertices $v_1$ and $v_2$ connected by an edge, and adds edges $v_1w_1$, $v_1w_2$ and $v_2w_3$, $v_2w_4$. We fix a planar embedding of the new graph by placing $v_1$, $v_2$ very close to where $v$ was drawn. We say that vertices $v_1$ and $v_2$ \emph{lie} on the same edges as the vertex $v$. Moreover, we call the edge $v_1v_2$ an \emph{auxiliary} edge. After splitting~$v$, every edge of $G_P$ has an edge corresponding to it, and every face of $G_P$ corresponds to a new one in a natural way. Additionally, the dual graph has one new edge, which is dual to~$v_1v_2$.

Let $G_S$ denote the \emph{split planarization} of the outer min-$k$-planar graph $G$, that is, the graph $G_P$ with all inner vertices split. See \cref{fig:spanning-trees} for an example. Observe that there is a one-to-one correspondence between $V(G_P^*)$ and $V(G_S^*)$. Further, every edge of $G_P^*$ has a corresponding edge of $G_S^*$. The following lemma shows how the properties of a spanning tree $T_P$ of $G_P$ and a spanning tree $\dual(T_P)$ of the dual graph $G_P^*$ are preserved after splitting the vertices of $G_P$.

\begin{lemma} \label{lem:good_st}
    Let $G$ be an expanded outer min-$k$-planar graph with its split planarization~$G_S$. Then there exists a spanning tree $T_S$ of $G_S$ and a spanning tree $T_S^* = \dual(T_S)$ of $G_S^*$ rooted at $f_o^*$, such that $\depth(f^*) \leq \floor{k/2}+1$ for every vertex $f^* \in V(G_S^*)$ and $E(T_S)$ contains all auxiliary edges of $G_S$.
\end{lemma}

\begin{proof}
    Let $G_P$ be a planarization of $G$. By Lemmas \ref{lem:bfs_tree} and \ref{lem:outerkpl_dual_dist}, there exists a spanning tree $T_P^*$ of $G_P^*$ that is rooted at $f_o^*$ whose vertices have depth at most $\floor{k/2}+1$. Let $T_P = \dual(T_P^*)$ denote the spanning tree of $G_P$.

    Let $G_S$ denote the graph obtained from $G_P$ by splitting each of its inner vertices. After this transformation, let $T_S^*$ be a tree constructed of edges corresponding to those of $T_P^*$. Clearly, $T_S^*$ is a spanning tree of the graph $G_S^*$. The spanning tree $T_S = \dual(T_S^*)$ of $G_S$ contains all auxiliary edges of $G_S$, because none of the duals of auxiliary edges are in $E(T_S^*)$, as they were not in $E(T_P^*)$.
\end{proof}

Now, we are ready to prove the main result of this section.

\begin{theorem} \label{thm:minkpl_upper}
    Let $G$ be an outer min-$k$-planar graph. Then $\tw(G) \leq 3 \cdot \floor{k/2} + 4$.
\end{theorem}

\begin{proof}
    By \cref{obs:expand_okpl} we may assume that $G$ is an expanded outer min-$k$-planar graph. Let $G_S$ be the split planarization of $G$. By \cref{lem:good_st} there exists a spanning tree $T_S$ of $G_S$ and a spanning tree $T_S^* = \dual(T_S)$ of $G_S^*$ rooted at $f_o^*$ such that $\depth(f^*) \leq \floor{k/2}+1$ for every vertex $f^* \in V(G_S^*)$ and $E(T_S)$ contains all auxiliary edges of $G_S$.

    \begin{figure}
        \begin{subfigure}[b]{0.47\textwidth}
            \centering
            \includegraphics[width=0.8\linewidth]{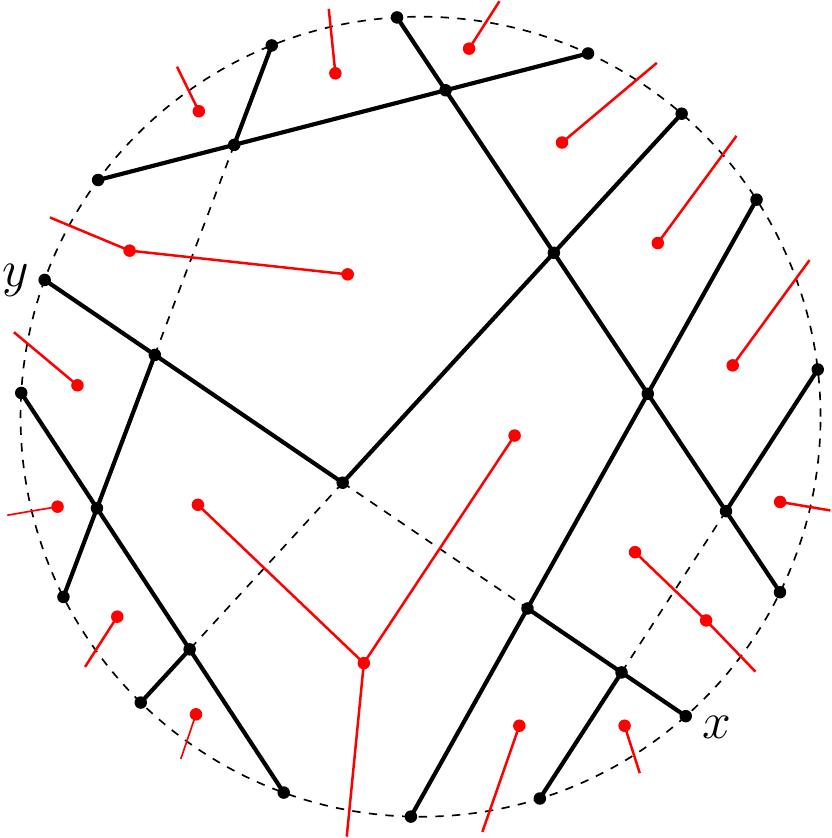}
            \subcaption[]{The graph $G_P$ with its spanning tree $T_P$ colored black and~a~spanning tree $T_P^*$ of $G_P^*$ in red.}
        \end{subfigure}
        \begin{subfigure}[b]{0.06\textwidth}
        \end{subfigure}
        \begin{subfigure}[b]{0.47\textwidth}
            \centering
            \includegraphics[width=0.8\linewidth]{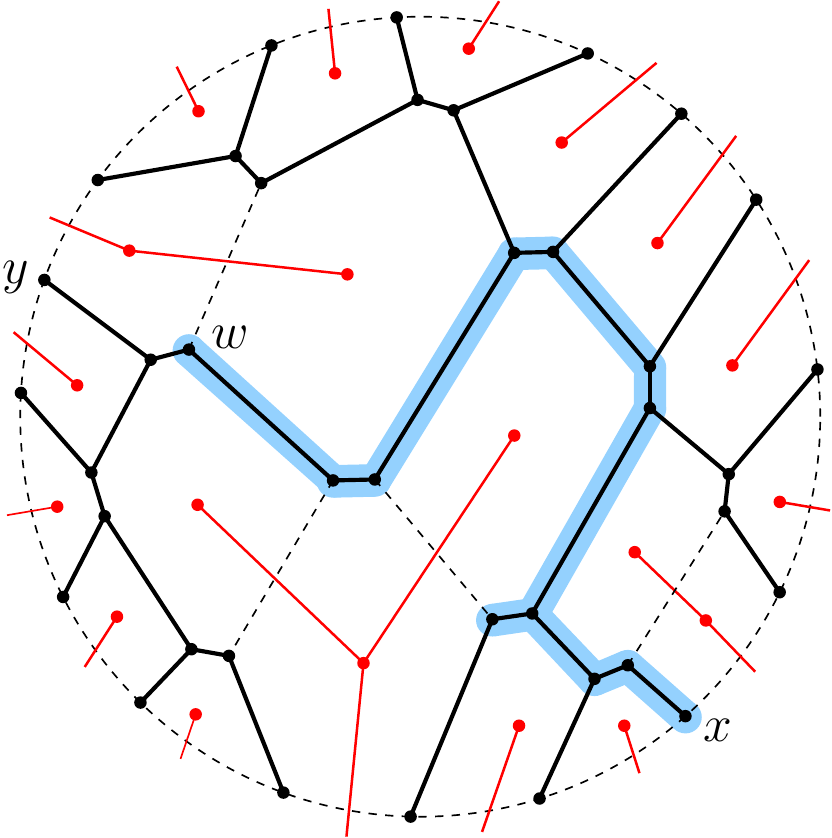}
            \captionsetup{textformat=simple}
            \subcaption[]{The graph $G_S$ with its spanning tree $T_S$ colored black and~a~spanning tree $T_S^*$ of $G_S^*$ in red.}
            \label{fig:spanning-trees-b}
        \end{subfigure}
        \caption{Drawings of an example graphs with their spanning trees and spanning trees of the dual graphs. An example walk, constructed in the proof of \cref{thm:minkpl_upper}, that is connecting $w$ and $x$, is marked in blue. The vertex $f_o^*$ is missing in both drawings.}
        \label{fig:spanning-trees}
    \end{figure}

    We orient the edges of $G$ as follows. Edges incident to the outer face are oriented clockwise, while all other edges are oriented arbitrarily. Observe that every vertex of $G$ has at most two incoming edges. 
    
    Now, we construct the tree decomposition $\mathcal{T} = (T_S, B)$ of the graph $G$. The bags of $\mathcal{T}$ are indexed by the vertices of the spanning tree $T_S$. We place vertices of $G$ into bags using the following rules.
    \begin{enumerate}
        \item For every outer vertex $v \in V(G_S)$, we place $v$ into the bag $B_v$.
        \item For every oriented edge $(x, y)$ of $G$, we place $x$ into the bag $B_y$.
        \item For every oriented edge $(x, y)$ of $G$, and for every inner vertex $v \in V(G_S)$ lying on $(x, y)$, we place $x$ into the bag~$B_v$.
        \item For every inner face $f$ and every edge $e^*$ on the path from $f^*$ to $f_o^*$ in $T_S^*$, where $e$ lies on the edge $(x, y)$ of $G$, we place $x$ into the bag $B_v$ for each vertex $v$ incident to~$f$.
    \end{enumerate}
    Note that, in rule 4, the edge $e$ is not in $E(T_S)$, so it is not an auxiliary edge, which implies that it is lying only on a single edge of $G$.
    
    For every edge $(x, y)$ of $G$, by rules 1 and 2, the bag $B_y$ contains both $x$ and $y$. Moreover, every vertex of $G$ is present in some bag. So, to prove that $\mathcal{T}$ is a tree decomposition of $G$, it suffices to show that for every vertex $x \in V(G)$, the set $\{w: x \in B_w\}$ induces a connected subtree in $T_S$. 

    Fix a vertex $x$ and an edge $(x,y)$ of $G$. Let $e = uv$ be an edge lying on $(x,y)$ such that $e \not\in E(T_S)$. Assume that $e^* = f_1^*f_2^*$ and $\depth_{T_S^*}(f_2^*) = \depth_{T_S^*}(f_1^*)+1$. Let $T_e^*$ be a subtree of $T_S^*$ induced by the set of all descendants of $f_2^*$, including $f_2^*$. Define $F_e = \set{f: f^* \in V(T_e^*)}$ and let $\boundary(F_e)$ be the set of vertices incident to some face in $F_e$. Observe that, by rule~4, $x$ is placed into all bags indexed by $\boundary(F_e)$. The set $\boundary(F_e)$ induces a connected subgraph of $T_S$, i.e.\ $T_S[\boundary(F_e)]$ is connected, containing both $u$ and $v$.
    Note that the bags of $\mathcal{T}$, into which we placed $x$ by rule 4, are exactly the bags of vertices of $T_S$ that are in $\boundary(F_e)$ for some edge $e\not\in E(T_S)$ lying on some edge $(x,y)$ of $G$.    

    By rules 1, 2 and 3, $x$ is contained in all bags $B_w$ such that $w$ lies on $(x,y)$ for some edge $(x,y)$ of $G$. We claim that the vertices indexing these bags, together with the vertices indexing bags we placed $x$ into by rule 4, form a connected subgraph of $T_S$. To see that, we show that for every vertex $w$ with $x \in B_w$, $w$ is connected to $x$ by a walk in $T_S$ such that bags indexed by the vertices of this walk contain $x$. 
    
    If $w$ lies on an edge $(x,y)$ of $G$ then, in order to construct this walk, we start at vertex $w$. We iterate over consecutive edges lying on $(x,y)$ between $w$ and $x$, starting at the edge incident to $w$. If the current edge $e$ is in $E(T_S)$, then we extend the walk by $e$. Otherwise, $e \not\in E(T_S)$.
    As $T_S[\boundary(F_e)]$ is connected and every bag of a vertex in $\boundary(F_e)$ contains $x$, we can extend the walk by some path in $T_S[\boundary(F_e)]$ connecting the endpoints of $e$.
    See \cref{fig:spanning-trees-b} for an example of such a walk.
    
    If $w$ is in the set $\boundary(F_e)$ for some edge $e$ lying on $(x,y)$, then we begin the walk with a path contained in $\boundary(F_e)$ between $w$ and an endpoint $v$ of $e$. We extend this walk by a walk between $v$ and $x$, whose existence we have already proven.

    Next, we bound the size of the bags in $\mathcal{T}$. Consider an inner vertex $v$ of $T_S$. It lies on exactly two edges of $G$, so by rule 3 we placed two vertices into $B_v$. Also, $v$ is incident to three non-outer faces of $G_S$. For every such face $f$ and every edge $e^*$ on the path from $f^*$ to $f_o^*$ in $T_S^*$, by rule 4 we placed one vertex into $B_v$. By \cref{lem:good_st} we have $\depth(f^*) \leq \floor{k/2}+1$, so each such path has at most $\floor{k/2}+1$ edges. Thus, $\norm{B_v} \leq 2 + 3 \cdot \left(\floor{k/2}+1\right)$. Now, consider an outer vertex $v$ of $T_S$. By rules 1 and 2, the bag $B_v$ contains $v$ and at most two other endpoints of edges incoming to $v$ in $G$. Also, $v$ is incident to two non-outer faces of~$G_S$. Hence, we derive a bound $\norm{B_v} \leq 3 + 2 \cdot \left(\floor{k/2}+1\right)$. The width of the constructed tree decomposition is at most \[\max\bigl\{2 + 3 \cdot \left(\floor{k/2}+1\right), 3 + 2 \cdot \left(\floor{k/2}+1\right)\bigr\}-1 = 2 + 3 \cdot \left(\floor{k/2}+1\right) -1= 3 \cdot \floor{k/2} + 4.\]

\end{proof}

\section{The separation number of outer min-$k$-planar graphs}
\label{sec:separators}

The inequality $\sn(G) \leq \tw(G)+1$, that bounds the separation number, holds for every graph~$G$. We remark that \cref{thm:minkpl_upper} directly implies that $\sn(G) \leq 3\cdot \floor{k/2}+5$ for every outer min-$k$-planar graph $G$.
By carefully choosing some bag $B_x$ of a tree decomposition, we can construct a balanced separation $(C,D)$ satisfying $C\cap D = B_x$.
To establish a better upper bound, we first prove a general lemma showing how we can obtain a balanced separation $(C,D)$ such that $C\cap D = B_x \cap B_y$ for some two neighboring vertices $x,y$ of a tree decomposition, which needs to satisfy some additional properties.

\begin{lemma} \label{lem:td_sep}
    Let $\mathcal{T} = (T, B)$ be a tree decomposition of a graph $G$. Assume that $\Delta(T) \leq 3$ and that every vertex $v \in V(G)$ is in at least two bags of $\mathcal{T}$. Let $a$ be an integer such that $\norm{B_x \cap B_y} \leq a$ for every edge $xy \in E(T)$. Then $G$ has a balanced separation of order at most~$a$.
\end{lemma}

\begin{proof}
    For every edge $xy \in E(T)$, after removing it from $T$, we obtain two connected components $C_x$ and $C_y$ of $T$ such that $x \in V(C_x)$ and $y \in V(C_y)$. We define $S_{x,y} = \bigcup_{v \in V(C_x)} B_v$ and $S_{y,x} = \bigcup_{v \in V(C_y)} B_v$. It is a well known fact that the pair $(S_{x,y}, S_{y,x})$ is a separation of $G$ of order $\norm{S_{x,y} \cap S_{y,x}} = \norm{B_x \cap B_y} \leq a$.

    We claim that there exists an edge $xy \in E(T)$ such that $(S_{x,y}, S_{y,x})$ is a balanced separation of $G$. Suppose the contrary. Then for every $xy \in E(T)$ it holds that $\norm{S_{x,y} \setminus S_{y,x}} > \frac{2}{3}n$ or $\norm{S_{y,x} \setminus S_{x,y}} > \frac{2}{3}n$, where $n = \norm{V(G)}$. Now, we orient every edge of $T$. If the first inequality holds then we orient $xy$ as $(y,x)$, in the other case as $(x, y)$. Also, note that  $\norm{S_{x,y} \setminus S_{y,x}} > \frac{2}{3}n$ implies $\norm{S_{y,x}} < \frac{1}{3}n$.
    
    The tree $T$ with oriented edges is an acyclic graph, so there exists a sink, that is, a vertex in $T$ such that all edges incident to $x$ are oriented towards $x$. Let $\set{y_1, \ldots, y_d}$, where $d \leq 3$, be the set of neighbors of $x$ in $T$. We have $\norm{S_{y_i, x}} < \frac{1}{3}n$. Moreover, \[\bigcup_{1 \leq i \leq d} S_{y_i, x} = \bigcup_{v \in V(T)\setminus \{x\}} B_v = V(G),\] since every vertex of $G$ is in at least two bags of $\mathcal{T}$. We obtain the following inequalities
    \[ \norm{V(G)} = \norm{\bigcup_{1 \leq i \leq d} S_{y_i, x}} \leq \sum_{1 \leq i \leq d} \norm{S_{y_i, x}} < d \cdot \frac{1}{3}n \leq n, \]
    which is a contradiction.
\end{proof}

Now, we are ready to establish an upper bound on the separation number of outer min-$k$-planar graphs.

\begin{theorem} \label{thm:sn_upper}
    Let $G$ be an outer min-$k$-planar graph. Then $\sn(G) \leq 2\cdot \floor{k/2} +4$.
\end{theorem}

\begin{proof}
    The class of outer min-$k$-planar graphs is closed under taking subgraphs. Therefore, it suffices to find a balanced separation of order at most $2\cdot \floor{k/2} +4$ for every maximal outer min-$k$-planar graph $G$. Let $H$ be an expanded outer min-$k$-planar graph obtained from $G$ by \cref{obs:expand_okpl}. By \cref{thm:minkpl_upper}, there exists a tree decomposition $\mathcal{T} = (T_S, B)$ of $H$, where $T_S$ is a spanning tree of the split planarization of $H$. From the proof of \cref{thm:minkpl_upper}, it follows that $\Delta(T_S) \leq 3$ and that every vertex $v \in V(H)$ is in at least two bags of $\mathcal{T}$ (since there is an oriented edge $(v, w)$ in $H$, so $v \in B_v$ and $v \in B_w$). 
    
    We construct a tree decomposition $\mathcal{T}' = (T_S, B')$ of $G$ with $B'_x = \set{\org(v): v \in B_x}$, where $\org(v)$ denotes the original vertex that $v$ replaced in the transformation described in \cref{obs:expand_okpl}. Every vertex $v \in V(G)$ is in at least two bags of $\mathcal{T}'$, since every image of $v$ is in at least two bags of $\mathcal{T}$. Every edge $vw \in E(G)$ is realized in some bag of $\mathcal{T}'$, because in $H$ there is an edge corresponding to $vw$ between an image of $v$ and an image of $w$. 
    To prove that, for every vertex $v$ of $G$, the bags of $\mathcal{T}'$ containing $v$ are spanning a connected subtree of $T_S$, we denote the images of $v$ by $v_1, \ldots, v_s$, ordered along the outer face. Since $H$ is maximal, for every $i \in \set{1, \ldots, s-1}$, there is an edge $v_iv_{i+1}$ in $E(H)$. Thus, the two subtrees of $T_S$ induced by the bags of $\mathcal{T}$ containing $v_i$ and those containing $v_{i+1}$ share a common vertex. Bags containing $v$ in $\mathcal{T}'$ are spanning a connected subtree of $T_S$, because this subtree is a union of subtrees spanned by the images of $v$. Hence, $\mathcal{T}'$ is a valid tree decomposition of $G$.

     We say that a vertex $v$ was placed into a bag $B_x$ of $\mathcal{T}$ due to rule 4 of constructing the tree decomposition being applicable to the vertex $v$ and a face $f$ if:
     \begin{itemize}
         \item the face $f$ is incident to $x$;
         \item there exists an edge $e^*$ on the path between $f_o^*$ and $f^*$ in $T_S^*$ such that $e$ lies on an edge $(v,w)$ of $H$, for some $w \in V(H)$.
     \end{itemize}
     Now, we want to show that, for every edge $xy \in E(T_S)$, we have $\norm{B_x' \cap B_y'} \leq 2\cdot \floor{k/2} +4$. Let $f_1$ and $f_2$ be the faces of $H_S$ incident to $xy$.
     
     \begin{claim*}
        If $v \in B_x' \cap B_y'$ then there exists an image $v_t$ of $v$ such that either
        \begin{itemize}
            \item $xy$ lies on an edge $(v_t, w)$ of $H$, for some $w \in V(H)$; or
            \item $v_t$ was placed into both $B_x$ and $B_y$ due to rule 4 of constructing $\mathcal{T}$ being applicable to $v_t$ and face $f_1$ or face $f_2$.
        \end{itemize}
    \end{claim*}

    \begin{claimproof}
        If $x$ has degree 3 in $H_S$, then let $f_x \not\in \set{f_1, f_2}$ be the face of $H_S$ incident to vertex~$x$. Similarly, if $y$ has degree 3 in $H_S$, then let $f_y \not\in \set{f_1, f_2}$ be the face of $H_S$ incident to $y$.
        Assume that $v \in B_x' \cap B_y'$, but no image of $v$ was placed into $B_x$ and $B_y$ due to the reasons stated in the claim. So there exist $v_{i_x}$ and $v_{i_y}$ that are, not necessarily distinct, images of $v$ such that $v_{i_x} \in B_x$ and $v_{i_y} \in B_y$. For $t \in \set{x, y}$, vertex $v_{i_t}$ was placed into $B_t$ because either
        \begin{enumerate}
            \item $t$ lies on an edge $(v_{i_t}, w_t)$ of $H$ such that $xy$ does not lie on $(v_{i_t}, w_t)$; or
            \item when $t$ has degree 3 in $H_S$, rule 4 of constructing $\mathcal{T}$ is applicable to the vertex $v_{i_t}$ and face $f_t$, i.e.~there exists an edge $(v_{i_t}, w_t)$ of $H$ and an edge $e_t$ of $H_S$ lying on $(v_{i_t}, w_t)$ such that $e_t^*$ is on the path between $f_t^*$ and $f_o^*$ in $T_S^*$. 
        \end{enumerate}

        Now, we draw a curve $\mathcal{C}$ on the drawing of $H_S$. The curve $\mathcal{C}$ consists of the drawing of the edge $xy$ and the drawing of an arc of the outer face between $v_{i_x}$ and $v_{i_y}$ that contains only images of $v$ (the images of $v$ are spanning a single arc of the outer face). Next, we add to $\mathcal{C}$ curves connecting $v_{i_x}$ with $x$ and $v_{i_y}$ with $y$. For $t \in \set{x, y}$, the way we draw these curves depends on the reason $v_{i_t}$ was placed into $B_t$, considered in the same order as above.
        \begin{enumerate}
            \item If $t$ lies on an edge $(v_{i_t}, w_t)$ of $H$, we draw along $(v_{i_t}, w_t)$, starting at vertex $v_{i_t}$ and ending at vertex $t$.
            \item If rule 4 of constructing $\mathcal{T}$ is applicable to the vertex $v_{i_t}$ and face $f_t$, let $p_t$ be a path between $f_t^*$ and $f_o^*$ in $T_S^*$. We draw along $(v_{i_t}, w_t)$, starting at vertex $v_{i_t}$ and ending at the crossing point with the drawing of $p_t$. We continue along $p_t$ till vertex $f_t^*$. Finally, we connect vertices $f_t^*$ and $t$ with a segment.
        \end{enumerate}
    
        \begin{figure}
            \centering
            \includegraphics[width=0.6\linewidth]{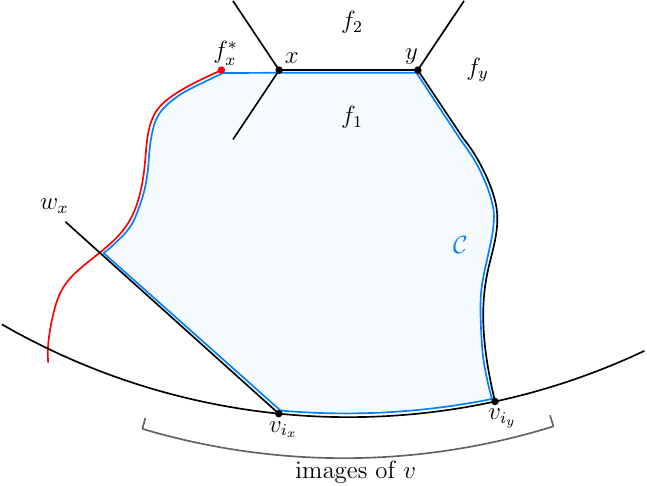}
            \caption{Drawing of an example curve $\mathcal{C}$.}
            \label{fig:curve}
        \end{figure}
    
        Note that if $x = v_{i_x}$ and $y = v_{i_y}$, then $\mathcal{C}$ is degenerated to the arc between $v_{i_x}$ and $v_{i_y}$, implying that $xy$ connects two consecutive images of $v$ -- contradiction. Otherwise, we claim that one of the closed regions induced by $\mathcal{C}$ contains $f_1$ or $f_2$. Indeed, $\mathcal{C}$ follows edges of $H_S$ and edges of $T_S^*$, but cannot contain $f_1^*$ nor $f_2^*$, because then rule 4 of constructing $\mathcal{T}$ would be applicable to vertex $v_{i_x}$ or $v_{i_y}$ and face $f_1$ or $f_2$. The segments between $f_t^*$ and $t$ does not intersect $f_1$ nor $f_2$.
        We may assume that $f_1$ is contained inside a closed region induced by~$\mathcal{C}$. Consider a path $p_1$ between $f_1^*$ and $f_o^*$ in $T_S^*$. Since $f_1^*$ is inside $\mathcal{C}$ and $f_o^*$ is outside $\mathcal{C}$, drawing of $p_1$ has to intersect $\mathcal{C}$. We consider where the first intersection point is located.
        \begin{itemize}
            \item Path $p_1$ cannot intersect $e$ nor the segments between $f_t^*$ and $t$.
            \item If $p_1$ intersects an edge $(v_{i_t}, w_t)$ then rule 4 of constructing $\mathcal{T}$ is applicable to $v_{i_t}$ and $f_1$.
            \item If $p_1$ intersects $p_t$, then $p_1$ follows along $p_t$ up to the intersection point with $(v_{i_t}, w_t)$, so the previous case applies.
            \item If $p_1$ intersects the arc of the outer face between $v_{i_x}$ and $v_{i_y}$ then it has to intersect an edge $(v_r, v_{r+1})$, where $v_r$ and $v_{r+1}$ are consecutive images of $v$ on the outer face. Hence, rule 4 of constructing $\mathcal{T}$ is applicable to $v_r$ and $f_1$.
        \end{itemize}
        In all cases, we obtain a contradiction.
    \end{claimproof}

    We proved that if $v \in B_x' \cap B_y'$ then there is an image $v_t$ of $v$ such that either
    \begin{itemize}
        \item $xy$ lies on an edge $(v_t, w)$ of $H$, for some $w \in V(H)$; or
        \item $v_t$ was placed into both $B_x$ and $B_y$ due to rule 4 of constructing the tree decomposition $\mathcal{T}$ being applicable to vertex $v_t$ and face $f_1$ or $f_2$.
    \end{itemize}
      Note that $xy$ lies on at most two edges of $H$ (two if $xy$ is an auxiliary edge, one otherwise). Moreover, the two paths from $f_o^*$ to $f_1^*$ and from $f_o^*$ to $f_2^*$ in $T_S^*$ each have at most $\floor{k/2}+1$ edges. Therefore, $\norm{B_x' \cap B_y'} \leq 2 + 2\cdot (\floor{k/2}+1) = 2\cdot \floor{k/2} +4$. By applying \cref{lem:td_sep} to $\mathcal{T}'$, we obtain that $G$ has a balanced separation of order at most $2\cdot \floor{k/2} +4$.
\end{proof}

To give a lower bound we define a graph called \emph{stacked prism}. A stacked prism $Y_{m, n}$ is an $m \times n$ grid with additional edges connecting the vertices of the first and the last row that are in the same column. The $Y_{m, n}$ has an outer $(2n-2)$-planar drawing, thus also an outer min-$(2n-2)$-planar drawing. In the cyclic order of the drawing, we place rows consecutively, one after another. The edges from rows cross no other edges and the edges from columns cross exactly $2n-2$ other edges.
The authors of~\cite{FirmanGKOW24} showed that for every number $n$ and for every sufficiently large even number $m$, $\sn(Y_{m, n}) = 2n$.
This leads to the following theorem.

\begin{theorem} \label{thm:sn_lower}
    For every even number $k$, there exists an outer min-$k$-planar graph $G$ such that $\sn(G) = k+2$.
\end{theorem}

We remark that the multiplicative constant of 1 in the upper bound given in \cref{thm:sn_upper} is tight, as it matches that of the lower bound in \cref{thm:sn_lower}.

\bibliography{biblio.bib}

\end{document}